\documentclass[conference]{IEEEtran}
\usepackage{ifpdf}
\usepackage{cite}
\usepackage{tikz}
\usepackage{lipsum}
\usepackage{mathtools}
\usepackage{cuted}
\usepackage[font=small]{caption}
\usepackage{balance}
\usepackage[left=0.63in,right=0.63in,top=0.67in,bottom=0.63in]{geometry}
\usetikzlibrary{decorations, decorations.text,}
\usepackage{float}
\usepackage{amsmath}
\usepackage{amsthm}
\usepackage{mathtools, cuted}
\usepackage{array}
\usepackage[caption=false,font=footnotesize]{subfig}
\usepackage{textcomp}
\usepackage{url}
\usepackage{paralist,esint,amssymb,booktabs,cases,bm,galois}

\newcommand\Se[1]{\mathcal{#1}}
\newcommand\Db[1]{\mathbb{#1}}

\newcommand\MB[1]{\left[#1\right]}
\newcommand\LB[1]{\{#1\}}

\newcommand{\RN}[1]{\textup{\uppercase\expandafter{\romannumeral#1}}}
\usepackage{cleveref}
\usetikzlibrary{calc}
\usetikzlibrary{shadings}
\newtheorem{theo}{Theorem}

\newtheorem{exam}{Example}
\newtheorem{defi}{Definition}

\newtheorem{cons}{Construction}

\usepackage{algorithm}
\usepackage[noend]{algpseudocode}
\makeatletter
\def\BState{\State\hskip-\ALG@thistlm}
\makeatother

\definecolor{apple green}{rgb}{0.17,0.75,0.13}

\definecolor{edit_ah}{rgb}{0.0,0.0,0.0}

\errorcontextlines\maxdimen
\makeatletter
\newcommand*{\algrule}[1][\algorithmicindent]{\makebox[#1][l]{\hspace*{.5em}\vrule height 0.9 \baselineskip depth 0.3\baselineskip}}%

\newcount\ALG@printindent@tempcnta
\def\ALG@printindent{%
    \ifnum \theALG@nested>0
        \ifx\ALG@text\ALG@x@notext
            \addvspace{0pt}
        \else
            \unskip
            \ALG@printindent@tempcnta=1
            \loop
                \algrule[\csname ALG@ind@\the\ALG@printindent@tempcnta\endcsname]%
                \advance \ALG@printindent@tempcnta 1
            \ifnum \ALG@printindent@tempcnta<\numexpr\theALG@nested+1\relax
            \repeat
        \fi
    \fi
    }%
\usepackage{etoolbox}
\patchcmd{\ALG@doentity}{\noindent\hskip\ALG@tlm}{\ALG@printindent}{}{\errmessage{failed to patch}}
\makeatother



\begin{document}
\bstctlcite{IEEEexample:BSTcontrol}
\title{\vspace{-0.35em}Topology-Aware Cooperative Data Protection in Blockchain-Based Decentralized Storage Networks}
\author{\IEEEauthorblockN{Siyi Yang$^1$, Ahmed Hareedy$^2$, Robert Calderbank$^2$, and Lara Dolecek$^1$}
\IEEEauthorblockA{$^1$ Electrical and Computer Engineering Department, University of California, Los Angeles, Los Angeles, CA 90095 USA\\
$^2$ Electrical and Computer Engineering Department, Duke University, Durham, NC 27708 USA\\
siyiyang@ucla.edu, ahmed.hareedy@duke.edu, robert.calderbank@duke.edu, and dolecek@ee.ucla.edu
}}
\maketitle

\begin{abstract}
From currency to cloud storage systems, the continuous rise of the blockchain technology is moving various information systems towards decentralization. Blockchain-based decentralized storage networks (DSNs) offer significantly higher privacy and lower costs to customers compared with centralized cloud storage associated with specific vendors. Coding is required in order to retrieve data stored on failing components. While coding solutions for centralized storage have been intensely studied, topology-aware coding for heterogeneous DSNs have not yet been discussed. In this paper, we propose a joint coding scheme where each node receives extra protection through the cooperation with nodes in its neighborhood in a heterogeneous DSN with any given topology. As an extension of, which also subsumes, our prior work on coding for centralized cloud storage, our proposed construction preserves desirable properties such as scalability and flexibility in networks with varying topologies. 

\begin{IEEEkeywords}
\color{edit_ah}Joint hierarchical coding, cooperative data protection, blockchain technology, decentralized storage networks.
\end{IEEEkeywords}
\end{abstract}

\IEEEpeerreviewmaketitle

\section{Introduction}
\label{sectoin: introduction}
The blockchain technology, first introduced by Satoshi Nakamoto as a technology supporting the digital currency called bitcoin, has been intensively discussed and regarded as a substantial innovation in cryptosystems \cite{nakamoto2019bitcoin,crosby2016blockchain,bagaria2018deconstructing}. Blockchain enables recording transactions through a decentralized deployment, which effectively addresses the potential issues of compromised data privacy and key abuse, arising from the existence of a \textcolor{edit_ah}{central} node that monopolizes all the actions and resource allocations in traditional centralized systems. \textcolor{edit_ah}{Decentralization has the potential to universally revolutionize a variety of applications, one of which is cloud storage.}

In contrast to traditional centralized storage based on the client-server model, where big companies monopolize renting the storage space to users, blockchain-based decentralized storage networks (DSNs) enable non-enterprise users to not only get access to the network storage space but also to contribute to increasing it via renting their remaining storage space on personal devices. Blockchain-based DSNs allocate the storage space and distribute the encrypted user data with the validation and integrity certification of a third-party, and thus have potential to offer higher privacy, higher reliability, and lower cost than currently available solutions. It has been found that the benefits in data integrity of decentralization are typically at the cost of higher latency and difficult maintenance.

While blockchain technology empowers DSNs to ensure the network-layer security, appropriate \textcolor{edit_ah}{error-correction codes (ECC) are needed to further improve the physical-layer reliability.} In a coded DSN, data are encoded via an ECC and stored in nodes, where each node is allowed to communicate and exchange information with \textcolor{edit_ah}{its} neighboring nodes. If data in each node are encoded independently, the system will be vulnerable to information leakage, data loss, and data tampering by malicious users. Therefore, enabling joint encoding of the data stored in all nodes such that nodes in the neighborhood cooperatively protect and validate their stored data collectively in the DSN is an essential requirement.


In existing literature on coding for distributed storage \cite{dimakis2010distributed,kong2010decentralized,ye2018cooperative}, \textcolor{edit_ah}{there has been no explicit consideration of random distributions and clustering nature of network nodes}. Clustered distributed storage has received attention in recent years in the context of multi-rack storage, where either the sizes of clusters and the capacities of the communication links are considered to be homogeneous \cite{tebbi2019multi,hou2019rack,chen2019explicit,prakash2018storage}, or the network topology \textcolor{edit_ah}{has a simple structure}\cite{li2010tree,pernas2013non}. However, \textcolor{edit_ah}{DSNs typically have more sophisticated topologies characterized by heterogeneity among capacities of communication links and erasure statistics of nodes due to the random and} dynamic nature of practical networks \cite{pernas2013non,wang2014heterogeneity,ibrahim2016green,sipos2018network,sipos2016erasure}. \textcolor{edit_ah}{Instead of solutions for simplified models, schemes that fit into any topology (referred to as ``topology-awareness'' later on) are desired to exploit the existing resources, achieving lower latency and higher reliability.} 

\textcolor{edit_ah}{To further reduce} latency and decoding complexity, we propose to provide each node with multiple ECC capabilities enabled by cooperating with neighboring nodes in a series of nested sets with increasing sizes: we refer to this as \textbf{ECC hierarchy} later on in the paper. As the size of the set increases, the associated ECC capability increases. In our method, we consider the case where each node stores encoded messages with local data protection, where the data length and codeword length can be customized by the user at each node. 

In this paper, we introduce a \textcolor{edit_ah}{topology-aware coding scheme that enables hierarchical cooperative data protection among nodes in a DSN, which is built upon our prior work in \cite{Yang2019HC} in the context of centralized cloud storage and preserves desirable properties like scalability and flexibility.} The rest of the paper is organized as follows. In \Cref{sec: model and prelim}, we introduce the DSN model and necessary preliminaries. In \Cref{sec: cooperative data protection}, we define ECC hierarchies as well as their depth, and present a coding scheme with depth $1$. In \Cref{sec: MLC}, we define the notion of compatible cooperation graphs, and propose an explicit construction of hierarchical codes for nodes with their cooperation graph being compatible. Finally, we summarize our results and discuss future directions in \Cref{section: conclusion}.

\section{System Model and Preliminaries}
\label{sec: model and prelim}

In this section, we discuss the model and mathematical representation of a DSN, as well as necessary preliminaries. Throughout the remainder of this paper, $\MB{N}$ refers to $\{1,2,\dots,N\}$. For vectors $\bold{u}$ and $\bold{v}$ of the same length $p$, $\bold{u}\succ\bold{v}$ and $\bold{u}\prec \bold{v}$ implies $u_i>v_i$ and $u_i<v_i$, for all $i\in\left[p\right]$, respectively; $\bold{u}\succeq \bold{v}$ means ``$\bold{u}\succ\bold{v}$ or $\bold{u}=\bold{v}$'', and $\bold{u}\preceq \bold{v}$ means ``$\bold{u}\prec\bold{v}$ or $\bold{u}=\bold{v}$''.

\begin{figure}
\centering
\includegraphics[width=0.35\textwidth]{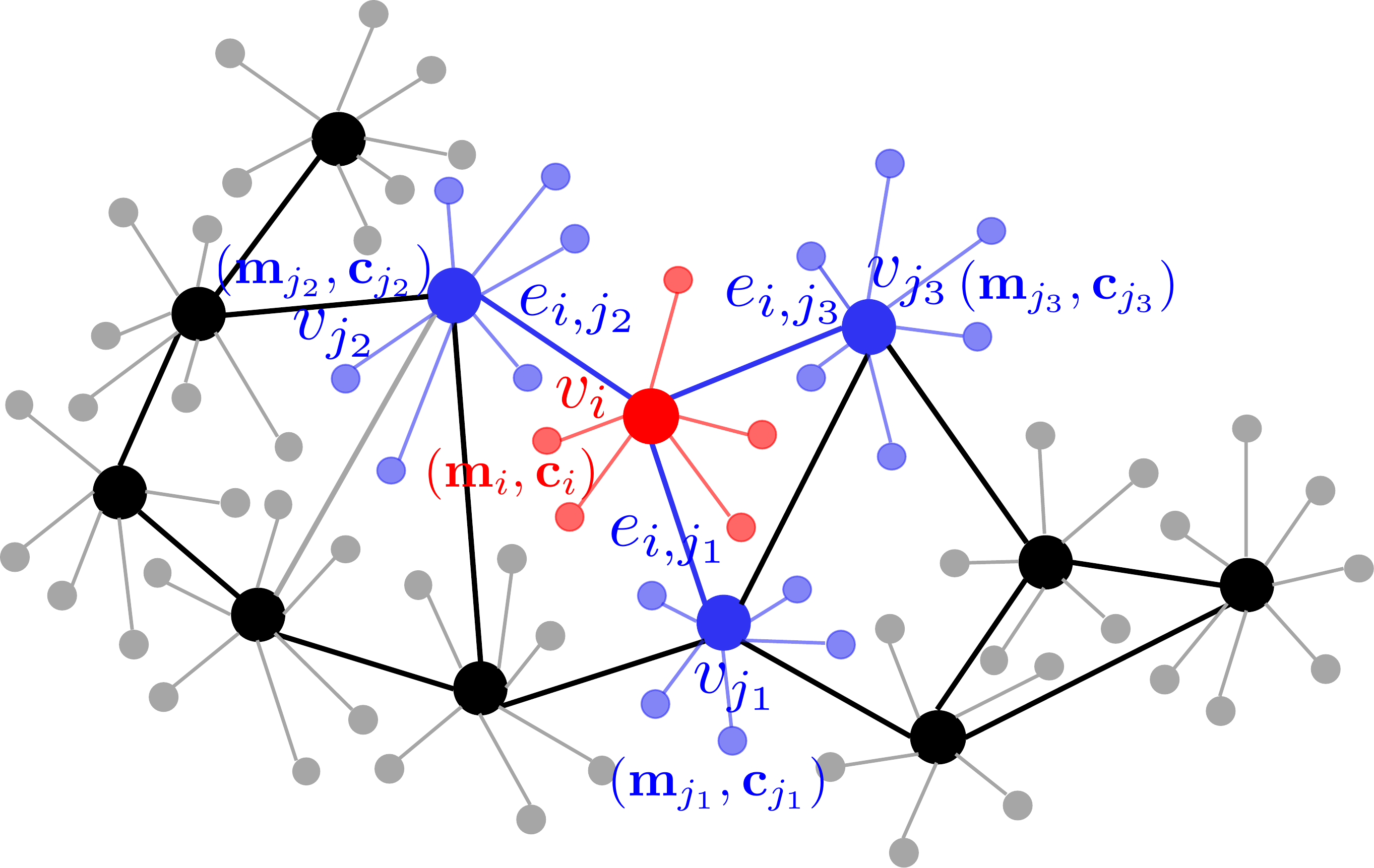}
\caption{Decentralized storage network (DSN). \textcolor{edit_ah}{For the cluster with the master node $v_i$}, message $\bold{m}_i$ is encoded to $\bold{c}_i$ \textcolor{edit_ah}{and symbols of $\bold{c}_i$ are stored distributively among non-master nodes that are locally connected to $v_i$. In the figures after Fig.~\ref{fig: figmodel}, we omit the local non-master nodes for clarity of figures.}}
\label{fig: figmodel}
\end{figure}

\subsection{Decentralized Network Storage}
\label{subsec: DNS model}
As shown in Fig.~\ref{fig: figmodel}, a DSN is modeled as a graph $G(V,E)$, where $V$ and $E$ denote the set of nodes (master only) and edges, respectively. \textcolor{edit_ah}{Codewords are stored among the nodes in a cluster. A failed node in a cluster is regarded as an erased symbol in the codeword stored in this cluster. A cluster is represented in $G$ by its master node $v_i\in V$ solely}. Each edge $e_{i,j}\in E$ represents a communication link connecting node $v_i$ and node $v_j$, through which $v_i$ and $v_j$ are allowed to exchange information. Denote the set of all neighbors of node $v_i$ by $\mathcal{N}_i$, e.g., $\mathcal{N}_i=\{j_1,j_2,j_3\}$ in Fig.~\ref{fig: figmodel}, and refer to it as the \textcolor{edit_ah}{\textbf{neighborhood}} of node $v_i$. Messages $\{\bold{m}_i\}_{v_i\in V}$ are jointly encoded as $\{\bold{c}_i\}_{v_i\in V}$, and $\bold{c}_i$ is stored in $v_i$. For a DSN denoted by $G(V,E)$, let $p=|V|$. Suppose $G$ is associated with a tuple $(\bold{n},\bold{k},\bold{r})\in\left(\mathbb{N}^p\right)^3$, where $\bold{k},\bold{r}\succ\bold{0}$ and $\bold{n}=\bold{k}+\bold{r}$. Note that $k_i$ represents the length of the message $\bold{m}_i$ associated with $v_i\in V$; $n_i$ and $r_i$ denote the length of $\bold{c}_i$ stored in $v_i$ and its syndrome, respectively.

\subsection{Preliminaries}
\label{subsec: prelim}
Based on the aforementioned notation, \textcolor{edit_ah}{a systematic generator matrix of a code} for $G(V,E)$ has the following structure:
\begin{equation}\label{eqn: GenMatDL}\bold{G}=\left[
\begin{array}{c|c|c|c|c|c|c}
\bold{I}_{k_1} & \bold{A}_{1,1} & \bold{0} & \bold{A}_{1,2} & \dots & \bold{0} & \bold{A}_{1,p}\\
\hline
\bold{0} & \bold{A}_{2,1} & \bold{I}_{k_2} & \bold{A}_{2,2} & \dots & \bold{0}& \bold{A}_{2,p}\\
\hline
\vdots & \vdots & \vdots & \vdots & \ddots & \vdots & \vdots \\
\hline
\bold{0} & \bold{A}_{p,1} & \bold{0} & \bold{A}_{p,2} & \dots & \bold{I}_{k_p} & \bold{A}_{p,p}\\
\end{array}\right],
\end{equation}
where all elements are from a Galois field $\textup{GF}(q)$, $q=2^\theta$ and $\theta \geq 2$. The major components of our construction are the so-called Cauchy matrices specified in \Cref{CauchyMatrix}.

\begin{defi} \emph{\textbf{(Cauchy matrix)}} \label{CauchyMatrix} Let $s,t\in\Db{N}$ and $\textup{GF}(q)$ be a finite field of size $q$. Suppose $a_1,\dots,a_s,b_1,\dots,b_t$ are $s \times t$ distinct elements in $\textup{GF}(q)$. The following matrix is known as a \textbf{Cauchy matrix},

\begin{equation*}\left[
\begin{array}{cccc}
\frac{1}{a_1-b_1} & \frac{1}{a_1-b_2} & \dots & \frac{1}{a_1-b_t}\\
\frac{1}{a_2-b_1} & \frac{1}{a_2-b_2} & \dots & \frac{1}{a_2-b_t}\\
\vdots & \vdots &\ddots & \vdots \\
\frac{1}{a_s-b_1} & \frac{1}{a_s-b_2} & \dots & \frac{1}{a_s-b_t}\\
\end{array}\right].
\label{defi: Cauchy matrix}
\end{equation*}
We denote this matrix by $\bold{Y}(a_1,\dots,a_s;b_1,\dots,b_t)$.
\end{defi}

\section{Cooperative Data Protection}
\label{sec: cooperative data protection}
In this section, we first mathematically describe the ECC hierarchy of a DSN and its depth, \textcolor{edit_ah}{which specifies the ECC capabilities of nodes while cooperating with different sets of other nodes.} We then propose a cooperation scheme where each node only cooperates with its single-hop neighbors.

\begin{figure}
\centering
\includegraphics[width=0.4\textwidth]{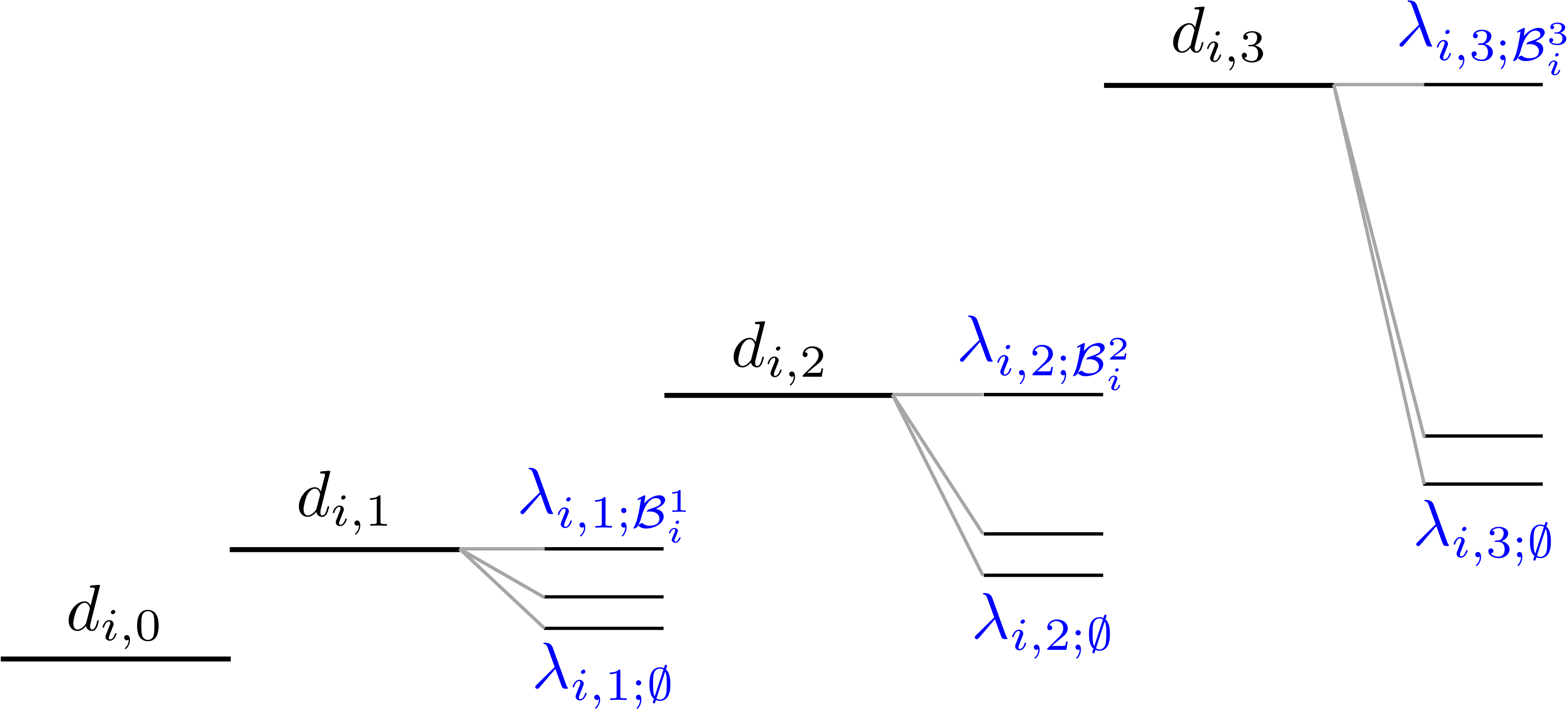}
\caption{ECC hierarchy of node $v_i\in V$.}
\label{fig: ECC Hierarchy}
\end{figure}

\begin{figure*}[!t]
\normalsize
\centering

\setcounter{equation}{2}
\setcounter{MaxMatrixCols}{12}
\begin{equation}\small
\scalebox{.7}{$
\begin{matrix}
\hline
\bold{A}_{1,1}&\bold{B}_{1,2}\bold{U}_2&\bold{0}&\bold{0}&\bold{0}&\bold{0}&\bold{0}&\bold{0}&\bold{0}&\bold{0}&\bold{0}&\bold{0}\\
\hline
\bold{B}_{2,1}\bold{U}_1&\bold{A}_{2,2}&\bold{B}_{2,3}\bold{U}_3&\bold{0}&\bold{B}_{2,5}\bold{U}_5&\bold{0}&\bold{0}&\bold{0}&\bold{0}&\bold{0}&\bold{0}&\bold{0}\\
\hline
\bold{0}&\bold{B}_{3,2}\bold{U}_2&\bold{A}_{3,3}&\bold{B}_{3,4}\bold{U}_4&\bold{0}&\bold{0}&\bold{0}&\bold{0}&\bold{0}&\bold{0}&\bold{0}&\bold{0}\\
\hline
\bold{0}&\bold{0}&\bold{B}_{4,3}\bold{U}_3&\bold{A}_{4,4}&\bold{B}_{4,5}\bold{U}_5&\bold{B}_{4,6}\bold{U}_6&\bold{0}&\bold{0}&\bold{0}&\bold{0}&\bold{0}&\bold{0}\\
\hline
\bold{0}&\bold{B}_{5,2}\bold{U}_2&\bold{0}&\bold{B}_{5,4}\bold{U}_4&\bold{A}_{5,5}&\bold{B}_{5,6}\bold{U}_6&\bold{0}&\bold{B}_{5,8}\bold{U}_8&\bold{0}&\bold{0}&\bold{0}&\bold{0}\\
\hline
\bold{0}&\bold{0}&\bold{0}&\bold{B}_{6,4}\bold{U}_4&\bold{B}_{6,5}\bold{U}_5&\bold{A}_{6,6}&\bold{B}_{6,7}\bold{U}_7&\bold{0}&\bold{0}&\bold{0}&\bold{0}&\bold{0}\\
\hline
\bold{0}&\bold{0}&\bold{0}&\bold{0}&\bold{0}&\bold{B}_{7,6}\bold{U}_6&\bold{A}_{7,7}&\bold{B}_{7,8}\bold{U}_8&\bold{B}_{7,9}\bold{U}_9&\bold{0}&\bold{B}_{7,11}\bold{U}_{11}&\bold{0}\\
\hline
\bold{0}&\bold{0}&\bold{0}&\bold{0}&\bold{B}_{8,5}\bold{U}_5&\bold{0}&\bold{B}_{8,7}\bold{U}_7&\bold{A}_{8,8}&\bold{B}_{8,9}\bold{U}_9&\bold{0}&\bold{0}&\bold{0}\\
\hline
\bold{0}&\bold{0}&\bold{0}&\bold{0}&\bold{0}&\bold{0}&\bold{B}_{9,7}\bold{U}_7&\bold{B}_{9,8}\bold{U}_8&\bold{A}_{9,9}&\bold{B}_{9,10}\bold{U}_{10}&\bold{0}&\bold{0}\\
\hline
\bold{0}&\bold{0}&\bold{0}&\bold{0}&\bold{0}&\bold{0}&\bold{0}&\bold{0}&\bold{B}_{10,9}\bold{U}_9&\bold{A}_{10,10}&\bold{B}_{10,11}\bold{U}_{11}&\bold{B}_{10,12}\bold{U}_{12}\\
\hline
\bold{0}&\bold{0}&\bold{0}&\bold{0}&\bold{0}&\bold{0}&\bold{B}_{11,7}\bold{U}_7&\bold{0}&\bold{0}&\bold{B}_{11,10}\bold{U}_{10}&\bold{A}_{11,11}&\bold{B}_{11,12}\bold{U}_{12}\\
\hline
\bold{0}&\bold{0}&\bold{0}&\bold{0}&\bold{0}&\bold{0}&\bold{0}&\bold{0}&\bold{0}&\bold{B}_{12,10}\bold{U}_{10}&\bold{B}_{12,11}\bold{U}_{11}&\bold{A}_{12,12}\\
\hline
\end{matrix}$}
\label{fig: genexample1}
\end{equation}
\hrulefill
\setcounter{equation}{1}
\end{figure*}

\subsection{ECC Hierarchy}
\label{subsec: ECC Hierarchy}

Denote the \textbf{ECC hierarchy} of node $v_i\in V$ by a sequence $\bold{d}_i=(d_{i,0},d_{i,1},\dots,d_{i,L_i})$, where $L_i$ is called the \textbf{depth} of $\bold{d}_i$, and $d_{i,l}$ represents the maximum number of erased symbols $v_i$ can recover in its codeword $\bold{c}_i$ from the $l$-th level cooperation, for all $l\in \MB{L_i}$. The $0$-th level cooperation refers to local erasure correction, i.e., the local node $v_i$ recovers its data without communicating with neighboring nodes.

For each $v_i\in V$ such that $L_i>0$, there exist two series of sets of nodes, denoted by $\varnothing\subset\Se{A}_i^1\subset\Se{A}_i^2\subset \dots \subset \Se{A}_i^{L_i}\subseteq V$ and $\{\Se{B}_i^l\}_{l=1}^{L_i}$, where $\Se{A}_i^l\cap \Se{B}_i^l=\varnothing$ for all $l\in \MB{L_i}$, and a sequence $\left(\lambda_{i,l;\Se{W}}\right)_{\varnothing\subseteq\Se{W}\subseteq\Se{B}_i^l}$. In the $l$-th level cooperation, node $v_i\in V$ tolerates $\lambda_{i,l;\Se{W}}$ \textcolor{edit_ah}{($\varnothing\subseteq\Se{W}\subseteq\Se{B}_i^l$)} erasures if all nodes in $\Se{A}_i^l\cup\Se{W}$ are able to decode their own messages, where the maximum value is $\lambda_{i,l;\Se{B}_i^l} = d_{i,l}$ and is reached when $\Se{W}=\Se{B}_i^l$; the minimum value is $\lambda_{i,l;\varnothing}$ and is reached when $\Se{W}=\varnothing$. See Fig.~\ref{fig: ECC Hierarchy}. We first take a look at the cooperation schemes with ECC hierarchy of depth $1$.


\subsection{Single-Level Cooperation}
\label{subsec: SLC}

We next discuss the case where each node only has cooperation of depth $1$. Consider a DSN represented by $G(V,E)$ that is associated with parameters $(\bold{n},\bold{k},\bold{r})$ and a class of sets $\{\Se{M}_i\}_{v_i\in V}$ such that $\varnothing\subset\Se{M}_i\subseteq \Se{N}_i$, for all $v_i\in V$. 
In \Cref{cons: 1}, we present a joint coding scheme where node $v_i$ only cooperates with nodes in $\Se{M}_i$, for all $v_i\in V$. Heterogeneity is obviously achieved since $n_i$, $k_i$, $r_i$, are not required to be identical for all $v_i\in V$. 

\begin{cons} \label{cons: 1} Let $G(V,E)$ represent a DSN associated with parameters $(\bold{n},\bold{k},\bold{r})$ and a local ECC parameter $\boldsymbol{\delta}$, where $\bold{r}\succ \boldsymbol{\delta}\succeq \bold{0}$. Let $p=|V|$ and $\textup{GF}(q)$ be a Galois field of size $q$, where $q>\max\limits_{v_i\in V} \left(n_i+\delta_i+\sum\nolimits_{j\in \Se{M}_i} \delta_j\right)$.   

For $i\in \MB{p}$, let $a_{i,x}$, $b_{i,y}$, $x\in\MB{k_i+\delta_i}$, and $y\in\MB{r_i+\sum\nolimits_{j\in \Se{M}_i} \delta_j}$ be distinct elements of $\textup{GF}(q)$.
Consider Cauchy matrix $\bold{T}_i\in \textup{GF}(q)^{(k_i+\delta_i)\times (r_i+\sum\nolimits_{j\in \Se{M}_i} \delta_j)}$ such that $\bold{T}_i=\bold{Y}(a_{i,1}, \dots, a_{i,k_i+\delta_i};b_{i,1},\dots,b_{i,r_i+\sum\nolimits_{j\in \Se{M}_i} \delta_j})$. Matrix $\bold{G}$ in (\ref{eqn: GenMatDL}) is assembled as follows. For $i\in\MB{p}$, obtain $\{\bold{B}_{i,j}\}_{j\in \Se{M}_i}$, $\bold{U}_i$, $\bold{A}_{i,i}$, according to the following partition of $\bold{T}_i$,

\vspace{-0.5em}\begin{equation}\label{eqn: CRS}
\bold{T}_i=\left[
\begin{array}{c|c}
\bold{A}_{i,i} & \begin{array}{c|c|c}
\bold{B}_{i,j_1} & \dots & \bold{B}_{i,j_{| \Se{M}_i|}}
\end{array}
\\
\hline
\bold{U}_i & \bold{Z}_{i}
\end{array}\right],
\end{equation}
where $\Se{M}_i=\{j_1,j_2,\dots,j_{| \Se{M}_i|}\}$, $\bold{A}_{i,i}\in \textup{GF}(q)^{k_i\times r_i}$,$\bold{U}_i\in \textup{GF}(q)^{\delta_i\times r_i}$, $\bold{B}_{i,j}\in \textup{GF}(q)^{k_i\times \delta_j}$ for $v_i\in V$ and $v_j\in \Se{M}_i$. Let $\bold{A}_{i,j}=\bold{B}_{i,j}\bold{U}_j$ if $v_j\in \Se{M}_i$, otherwise let it be a zero matrix. 
Denote the code with generator matrix $\bold{G}$ by $\Se{C}_1$.
\end{cons}

\begin{theo} \label{theo: ECCcons1} In a DSN with $\Se{C}_1$, $\bold{d}_i=(r_i-\delta_i,r_i+\sum\nolimits_{j\in \Se{M}_i}\delta_j)$, $\Se{A}_i^1=\Se{M}_i$ and $\textcolor{edit_ah}{\Se{B}_i^1=\bigcup\nolimits_{v_j\in\Se{M}_i}\left(\Se{M}_j\setminus(\{v_i\}\cup \Se{M}_i)\right)}$, for all $v_i\in V$. The ECC hierarchy associated with $d_{i,1}$ is $(\lambda_{i,1;\Se{W}})_{\varnothing\subseteq\Se{W}\subseteq \Se{B}_i^1}$, where $\lambda_{i,1;\Se{W}}=r_i+\sum\nolimits_{v_j\in \Se{M}_i,(\Se{M}_j\setminus\{v_i\})\subseteq (\Se{M}_i\cup\Se{W})}\delta_j$.
\end{theo}

\begin{proof} The local ECC capability $d_{i,0}=r_i-\delta_i$ at node $v_i$ has been proved in \cite[Construction 1]{Yang2019HC}. Therefore, we only need to prove the statement associated with the $1$-st level cooperation, and we discuss it in \Cref{exam: exam1}.

We first notice that for any $v_i\in V$, $v_j\in \Se{M}_i$, if $\bold{m}_j$ is recoverable, then since $\bold{A}_{j,j}$ and $\bold{U}_j$ have linearly independent columns, the sum $\bold{s}_{j}$ of cross parities $\bold{m}_{j'}\bold{B}_{j',j}$ generated because of $\{\bold{B}_{j',j}\bold{U}_{j}\}_{v_{j'}\in\Se{M}_j}$ can be calculated. If all the messages $\{\bold{m}_{j'}\}_{v_{j'}\in\Se{M}_j\setminus\{v_i\}}$ are further recoverable, then the cross parities $\bold{m}_i\bold{B}_{i,j}$ generated because of matrix $\bold{B}_{i,j}\bold{U}_j$ can be computed from $\bold{m}_i\bold{B}_{i,j}=\bold{s}_{j}-\sum\nolimits_{v_{j'}\in\Se{M}_j\setminus\{v_i\}}\bold{m}_{j'}\bold{B}_{j',j} $.  

Previous discussion implies that for any $\varnothing\subseteq\Se{W}\subseteq \Se{B}_i^1$, if $(\Se{M}_j\setminus\{v_i\})\subseteq (\Se{M}_i\cup\Se{W})$, then the additional $\delta_j$ cross parities $\bold{m}_i\bold{B}_{i,j}$ for $\bold{m}_i$ can be obtained. Therefore, $\lambda_{i,1;\Se{W}}=r_i+\sum\nolimits_{v_j\in \Se{M}_i,(\Se{M}_j\setminus\{v_i\})\subseteq (\Se{M}_i\cup\Se{W})}\delta_j$ and $d_{i,1}=\lambda_{i,1;\Se{B}_i^1}=r_i+\sum\nolimits_{j\in \Se{M}_i}\delta_j$.
\end{proof}

\begin{figure}
\centering
\includegraphics[width=0.35\textwidth]{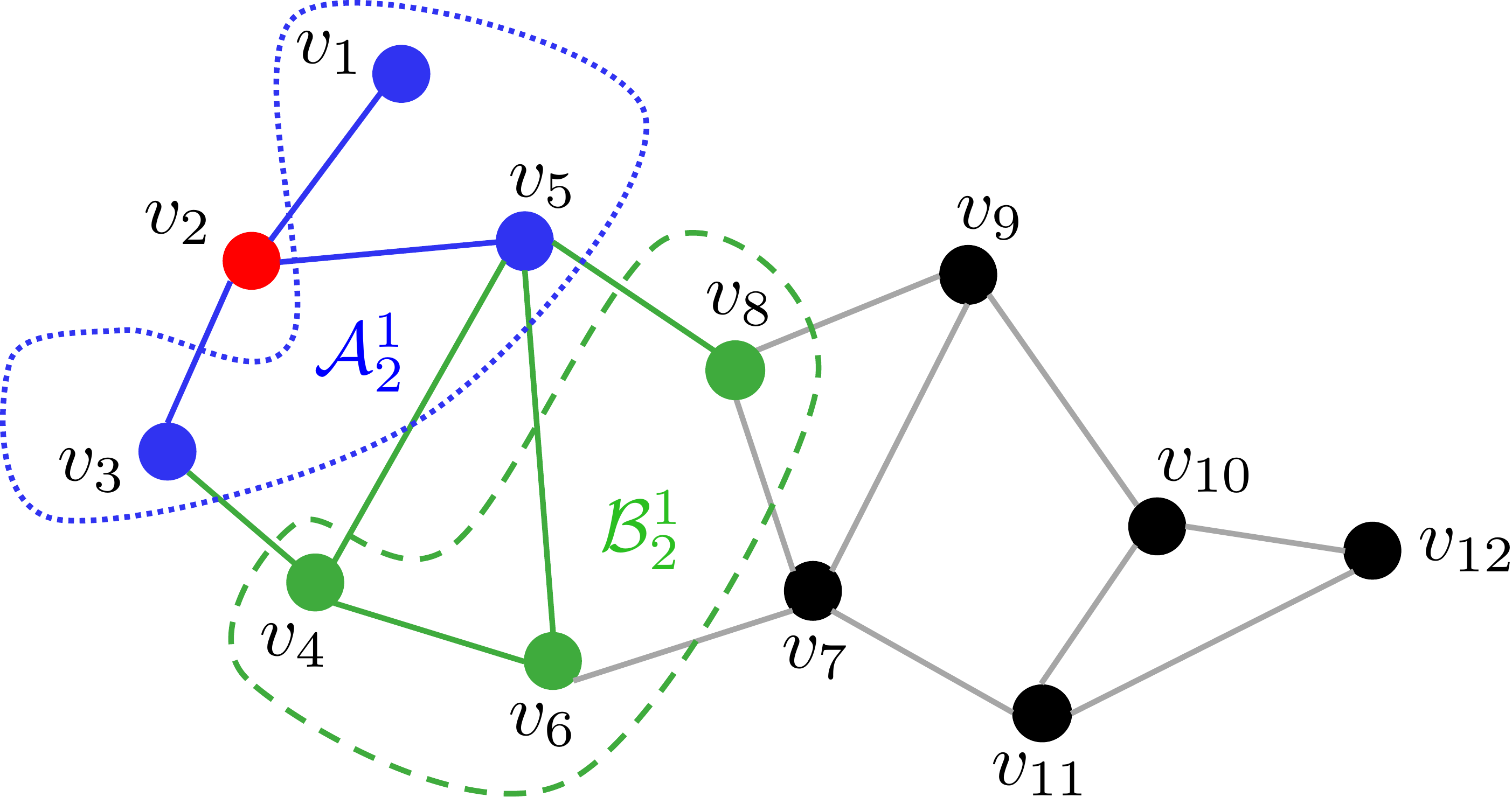}
\caption{DSN for \Cref{exam: exam1}.}
\label{fig: example1}
\end{figure}

\begin{exam} \label{exam: exam1} Consider the DSN shown in Fig.~\ref{fig: example1}, let $\Se{M}_i=\Se{N}_i$ in \Cref{cons: 1}, for all $i\in \MB{12}$. The matrix in (\ref{fig: genexample1}) is obtained by removing all the block columns of identity surrounded by zero matrices from the generator matrix (\ref{eqn: GenMatDL}) of $\Se{C}_1$, and is referred to as the \textbf{non-systematic component} of the generator matrix. 

Take node $v_2$ as an example. Observe that $\Se{A}_2^1=\Se{M}_2=\{v_1,v_3,v_5\}$, $\Se{A}_1^1=\Se{M}_1=\{v_2\}$, $\Se{A}_3^1=\Se{M}_3=\{v_2,v_4\}$, and $\Se{A}_5^1=\Se{M}_5=\{v_2,v_4,v_6,v_8\}$. Therefore, $\Se{B}_2^1=\bigcup\nolimits_{j\in\{1,3,5\}} \Se{M}_j\setminus\{v_1,v_2,v_3,v_5\}=\{v_4,v_6,v_8\}$. Moreover, $\textcolor{edit_ah}{\bold{d}_{2}}=(r_2-\delta_2,r_2+\sum\nolimits_{j\in\{1,3,5\}}\delta_j)$, $\lambda_{2,1;\varnothing}=\lambda_{2,1;\{v_6\}}=\lambda_{2,1;\{v_8\}}=\lambda_{2,1;\{v_6,v_8\}}=r_2+\delta_1$, $\lambda_{2,1;\{v_4\}}=\lambda_{2,1;\{v_4,v_6\}}=\lambda_{2,1;\{v_4,v_8\}}=r_2+\delta_1+\delta_3$, and $\lambda_{2,1;\{v_4,v_6,v_8\}}=r_2+\delta_1+\delta_3+\delta_5$.

Consider the case where the $1$-st level cooperation of $v_2$ is initiated, i.e, the number of erasures lies within the interval $\left[r_2-\delta_2+1,r_2+\delta_1+\delta_3+\delta_5\right]$. Then, if $\bold{m}_1,\bold{m}_3,\bold{m}_5$ are all locally recoverable, the cross parities $\bold{m}_1\bold{B}_{1,2}$, $\bold{m}_3\bold{B}_{3,2}$, $\bold{m}_5\bold{B}_{5,2}$ computed from the non-diagonal parts in the generator matrix can be subtracted from the parity part of $\bold{c}_2$ to get $\bold{m}_2\bold{A}_{2,2}$. Moreover, the successful decoding of $\bold{m}_1$ makes $\bold{m}_{2}\bold{B}_{2,1}$ known to $v_2$. This process provides $(r_2+\delta_1)$ parities for $\bold{m}_2$, and thus allows $v_2$ to tolerate $(r_2+\delta_1)$ erasures.

In order to correct more than $(r_2+\delta_1)$ erasures, we need extra cross parities generated from $\bold{B}_{2,3}\bold{U}_3$ and $\bold{B}_{2,5}\bold{U}_5$. However, local decoding only allows $v_3$, $v_5$ to know $\bold{m}_{2}\bold{B}_{2,3}+\bold{m}_{4}\bold{B}_{4,3}$ and $\textcolor{edit_ah}{\bold{m}_{2}\bold{B}_{2,5}+\bold{m}_{4}\bold{B}_{4,5}+\bold{m}_{6}\bold{B}_{6,5}+\bold{m}_{8}\bold{B}_{8,5}}$, respectively. Therefore, $v_3$ needs $\bold{m}_4$ to be recoverable to obtain the extra $\delta_3$ cross parities, and $v_5$ needs \textcolor{edit_ah}{$\bold{m}_4$,} $\bold{m}_6$, $\bold{m}_8$ to be recoverable to obtain the extra $\delta_5$ cross parities. 

\end{exam}

\textcolor{edit_ah}{As shown in \Cref{exam: exam1}, instead of presenting a rigid~ECC capability, our proposed scheme enables nodes to have correction of a growing number of erasures with bigger sets of neighboring nodes recovering their messages. Therefore, nodes automatically choose the shortest path to recover their messages, significantly increasing the average recovery speed, especially when the erasures are distributed non-uniformly and sparsely, which is important for blockchain-based DSNs \cite{zhu2019blockchain,underwood2016blockchain}. Moreover, nodes with higher reliabilities are utilized to help decode the data of less reliable nodes, enabling correction of error patterns that are not recoverable in our previous work in \cite{Yang2019HC}. We show these properties in \Cref{exam: speed} and \Cref{exam: erasure pattern}.}

\begin{figure}
\centering
\includegraphics[width=0.3\textwidth]{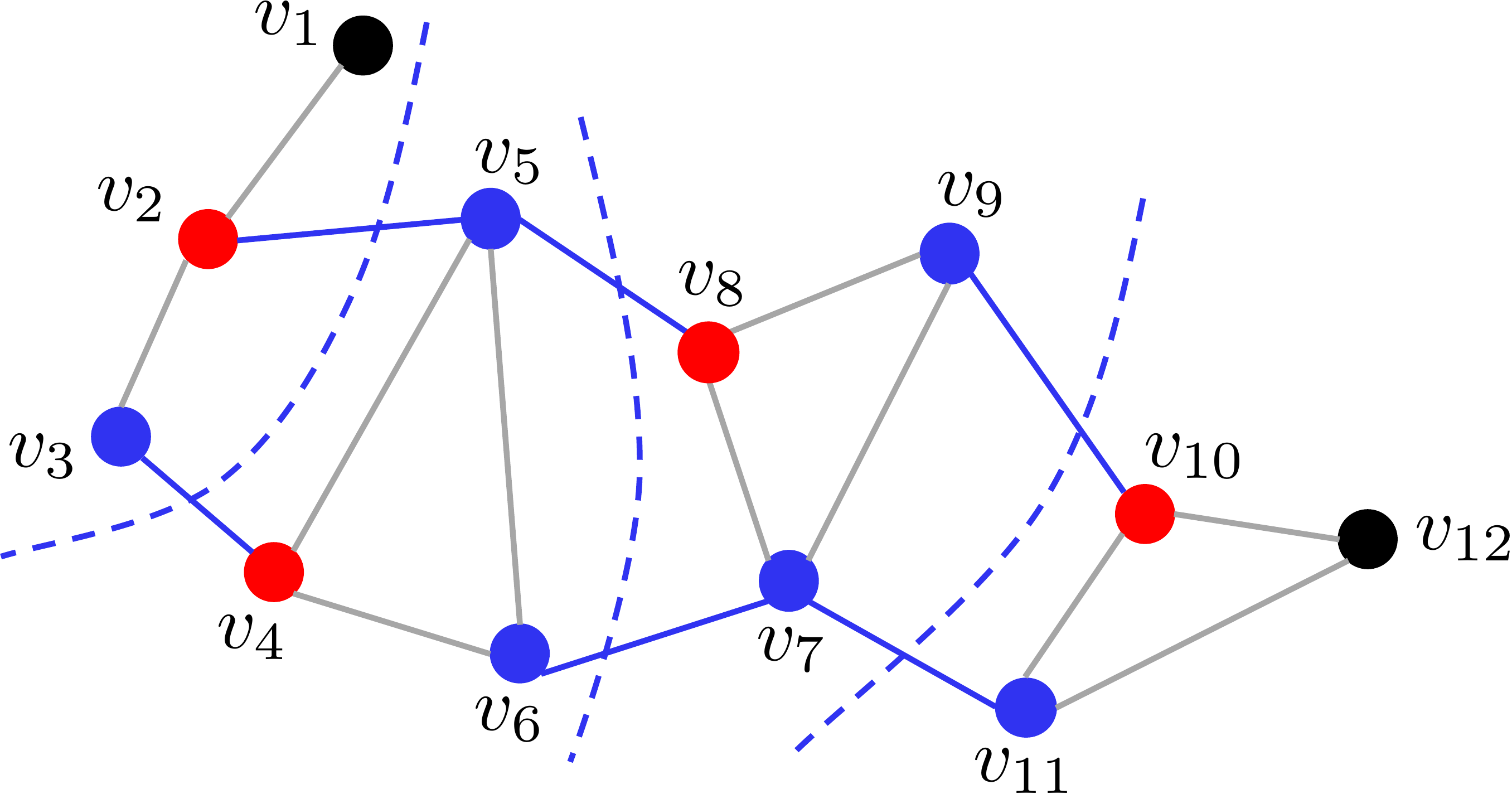}
\caption{The erasure pattern in \Cref{exam: erasure pattern}.}
\label{fig: compare}
\end{figure}

\begin{exam} \label{exam: speed} \emph{\textbf{(Faster Recovery Speed)}} Consider the DSN with the cooperation scheme specified in \Cref{exam: exam1}. Suppose the time to be consumed on transferring information through the communication link $e_{i,j}$ is $t_{i,j}\in\Db{R}^{+}$, where $\textcolor{edit_ah}{t_{i,j}=t_{j,i}}$ for all $i,j\in\MB{12}$, $i\neq j$, and $\max\{t_{1,2},t_{2,5}\}<(t_{2,3}+t_{3,4})<t_{2,5}+\min\{t_{4,5},t_{5,6},t_{5,8}\}$.

Consider the case where node $\bold{c}_2$ has $(r_2+1)$ erasures, which implies that apart from the case of $\bold{m}_1$, $\bold{m}_3$, $\bold{m}_5$ being obtained locally, recovering $\bold{m}_4$ is sufficient for $v_2$ to successfully obtain its message. The time consumed for decoding is $(t_{2,3}+t_{3,4})$. Therefore, any system using network coding with the property that a node failure is recovered through accessing more than $4$ other nodes will need longer processing time for this case. 
\end{exam}

\begin{exam} \label{exam: erasure pattern} \emph{\textbf{(Flexible Erasure Patterns)}} Consider the DSN with the cooperation scheme specified in \Cref{exam: exam1}. Suppose $\{\bold{m}_i\}_{i\notin\{2,4,8,10\}}$ are all locally recoverable. Then, consider the case where $\bold{m}_i$ has $(r_i+1)$ erasures for $i\in\{2,4,8,10\}$, which exemplifies a correctable erasure pattern for our proposed codes. 

The hierarchical coding scheme presented in \cite{Yang2019HC} can recover from this erasure pattern only if the code used adopts a partition of all nodes into $4$ disjoint groups, each of which contains exactly a node from $\{v_2,v_4,v_8,v_{10}\}$, as shown in Fig.~\ref{fig: compare}. Moreover, the partition of the code in \cite{Yang2019HC} results in a reduction of the ECC capability of the $1$-st level cooperation at every node except for $v_1,v_{12}$ because the additional information originally flow through edges marked in blue no longer exist.
\end{exam}

\section{Multi-Level Cooperation}
\label{sec: MLC}

In this section, we extend the construction presented in \Cref{subsec: SLC} to codes with ECC hierarchies of depth larger than $1$. \textcolor{edit_ah}{We first define the so-called \textbf{cooperation graphs} that describe how the nodes are coupled to cooperatively transmit information, and then prove the existence of hierarchical codes over a special class of cooperation graphs: the so-called \textbf{compatible graphs}.  }

\subsection{Cooperation Graphs}
\label{subsec: cooperation graphs}

Based on the aforementioned notation, for each $v_i\in V$ and $l\in\MB{L_i}$, let $\Se{I}^{l}_i=\Se{A}^{l}_i\setminus\Se{A}^{l-1}_i$ and refer to it as the $l$-th \textbf{helper} of $v_i$. 
We next define the so-called \textbf{cooperation matrix}.

\begin{defi} \label{defi: cooperation matrix} For a joint coding scheme $\Se{C}$ for a DSN represented by $G(V,E)$ with $|V|=p$, the matrix $\bold{D}\in\mathbb{N}^{p\times p}$, in which $\bold{D}_{i,j}$ equals to $l$ for all $i,j\in\MB{p}$ such that $j\in\mathcal{I}^{l}_i$, and zero otherwise, is called the \textbf{cooperation matrix}.
\end{defi}


As an example, the cooperation matrix in \Cref{exam: exam1} is exactly the adjacency matrix of the graph in Fig.~\ref{fig: example1}. Note that not every matrix is a cooperation matrix of a~set of joint coding schemes. In \Cref{subsec: construction}, we prove the existence of codes if the cooperation matrix represents a so-called \textbf{compatible graph}. Before going into details of the construction, we look at an example to obtain some intuition.
\begin{figure}
\centering
\includegraphics[width=0.45\textwidth]{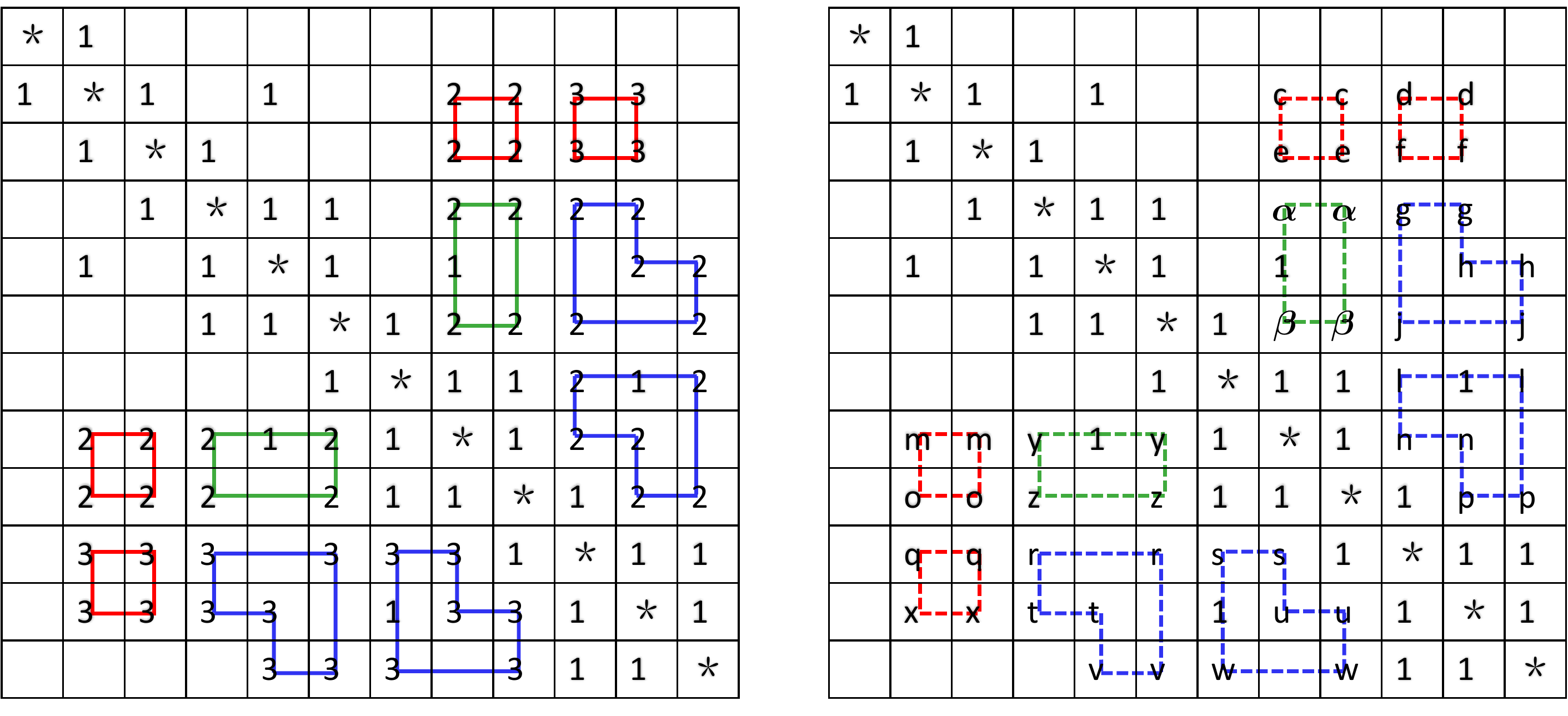}
\caption{Matrices $\bold{D}$ (left) and $\bold{X}$ (right) in \Cref{exam: exam3}.}
\label{fig: CoopMatrix}
\end{figure}

\begin{figure*}[!t]
\normalsize
\centering

\setcounter{equation}{3}
\setcounter{MaxMatrixCols}{12}
\begin{equation}\small
\scalebox{.7}{$
\begin{matrix}
\hline
\bold{A}_{1,1}&\bold{B}_{1,2}\bold{U}_2&\bold{0}&\bold{0}&\bold{0}&\bold{0}&\bold{0}&\bold{0}&\bold{0}&\bold{0}&\bold{0}&\bold{0}\\
\hline
\bold{B}_{2,1}\bold{U}_1&\bold{A}_{2,2}&\bold{B}_{2,3}\bold{U}_3&\bold{0}&\bold{B}_{2,5}\bold{U}_5&\bold{0}&\bold{0}&\textcolor{blue}{\bold{B}_c\bold{V}_{8;2}}&\textcolor{blue}{\bold{B}_c\bold{V}_{9;2}}&\textcolor{red}{\bold{B}_d\bold{V}_{10;3}}&\textcolor{red}{\bold{B}_d\bold{V}_{11;3}}&\bold{0}\\
\hline
\bold{0}&\bold{B}_{3,2}\bold{U}_2&\bold{A}_{3,3}&\bold{B}_{3,4}\bold{U}_4&\bold{0}&\bold{0}&\bold{0}&\textcolor{blue}{\bold{B}_e\bold{V}_{8;2}}&\textcolor{blue}{\bold{B}_e\bold{V}_{9;2}}&\textcolor{red}{\bold{B}_f\bold{V}_{10;3}}&\textcolor{red}{\bold{B}_f\bold{V}_{11;3}}&\bold{0}\\
\hline
\bold{0}&\bold{0}&\bold{B}_{4,3}\bold{U}_3&\bold{A}_{4,4}&\bold{B}_{4,5}\bold{U}_5&\bold{B}_{4,6}\bold{U}_6&\bold{0}&\textcolor{apple green}{\bold{B}_{\alpha}\bold{V}_{8;2}}&\textcolor{apple green}{\bold{B}_{\alpha}\bold{V}_{9;2}}&\textcolor{blue}{\bold{B}_g\bold{V}_{10;2}}&\textcolor{blue}{\bold{B}_g\bold{V}_{11;2}}&\bold{0}\\
\hline
\bold{0}&\bold{B}_{5,2}\bold{U}_2&\bold{0}&\bold{B}_{5,4}\bold{U}_4&\bold{A}_{5,5}&\bold{B}_{5,6}\bold{U}_6&\bold{0}&\bold{B}_{5,8}\bold{U}_8&\bold{0}&\textcolor{blue}{\bold{B}_h\bold{V}_{10;2}}&\bold{0}&\textcolor{blue}{\bold{B}_h\bold{V}_{12;2}}\\
\hline
\bold{0}&\bold{0}&\bold{0}&\bold{B}_{6,4}\bold{U}_4&\bold{B}_{6,5}\bold{U}_5&\bold{A}_{6,6}&\bold{B}_{6,7}\bold{U}_7&\textcolor{apple green}{\bold{B}_{\beta}\bold{V}_{8;2}}&\textcolor{apple green}{\bold{B}_{\beta}\bold{V}_{9;2}}&\textcolor{blue}{\bold{B}_j\bold{V}_{10;2}}&\bold{0}&\textcolor{blue}{\bold{B}_j\bold{V}_{12;2}}\\
\hline
\bold{0}&\bold{0}&\bold{0}&\bold{0}&\bold{0}&\bold{B}_{7,6}\bold{U}_6&\bold{A}_{7,7}&\bold{B}_{7,8}\bold{U}_8&\bold{B}_{7,9}\bold{U}_9&\textcolor{blue}{\bold{B}_l\bold{V}_{10;2}}&\bold{B}_{7,11}\bold{U}_{11}&\textcolor{blue}{\bold{B}_l\bold{V}_{12;2}}\\
\hline
\bold{0}&\textcolor{blue}{\bold{B}_m\bold{V}_{2;2}}&\textcolor{blue}{\bold{B}_m\bold{V}_{3;2}}&\textcolor{apple green}{\bold{B}_{y}\bold{V}_{4;2}}&\bold{B}_{8,5}\bold{U}_5&\textcolor{apple green}{\bold{B}_{y}\bold{V}_{6;2}}&\bold{B}_{8,7}\bold{U}_7&\bold{A}_{8,8}&\bold{B}_{8,9}\bold{U}_9&\textcolor{blue}{\bold{B}_n\bold{V}_{10;2}}&\textcolor{blue}{\bold{B}_n\bold{V}_{11;2}}&\bold{0}\\
\hline
\bold{0}&\textcolor{blue}{\bold{B}_o\bold{V}_{2;2}}&\textcolor{blue}{\bold{B}_o\bold{V}_{3;2}}&\textcolor{apple green}{\bold{B}_{z}\bold{V}_{4;2}}&\bold{0}&\textcolor{apple green}{\bold{B}_{z}\bold{V}_{6;2}}&\bold{B}_{9,7}\bold{U}_7&\bold{B}_{9,8}\bold{U}_8&\bold{A}_{9,9}&\bold{B}_{9,10}\bold{U}_{10}&\textcolor{blue}{\bold{B}_p\bold{V}_{11;2}}&\textcolor{blue}{\bold{B}_p\bold{V}_{12;2}}\\
\hline
\bold{0}&\textcolor{red}{\bold{B}_q\bold{V}_{2;3}}&\textcolor{red}{\bold{B}_q\bold{V}_{3;3}}&\textcolor{red}{\bold{B}_r\bold{V}_{4;2}}&\bold{0}&\textcolor{red}{\bold{B}_r\bold{V}_{6;3}}&\textcolor{red}{\bold{B}_s\bold{V}_{7;3}}&\textcolor{red}{\bold{B}_s\bold{V}_{8;3}}&\bold{B}_{10,9}\bold{U}_9&\bold{A}_{10,10}&\bold{B}_{10,11}\bold{U}_{11}&\bold{B}_{10,12}\bold{U}_{12}\\
\hline
\bold{0}&\textcolor{red}{\bold{B}_x\bold{V}_{2;3}}&\textcolor{red}{\bold{B}_x\bold{V}_{3;3}}&\textcolor{red}{\bold{B}_t\bold{V}_{4;2}}&\textcolor{red}{\bold{B}_t\bold{V}_{5;2}}&\bold{0}&\bold{B}_{11,7}\bold{U}_7&\textcolor{red}{\bold{B}_u\bold{V}_{8;3}}&\textcolor{red}{\bold{B}_u\bold{V}_{9;3}}&\bold{B}_{11,10}\bold{U}_{10}&\bold{A}_{11,11}&\bold{B}_{11,12}\bold{U}_{12}\\
\hline
\bold{0}&\bold{0}&\bold{0}&\bold{0}&\textcolor{red}{\bold{B}_v\bold{V}_{5;2}}&\textcolor{red}{\bold{B}_v\bold{V}_{6;3}}&\textcolor{red}{\bold{B}_w\bold{V}_{7;3}}&\bold{0}&\textcolor{red}{\bold{B}_w\bold{V}_{9;3}}&\bold{B}_{12,10}\bold{U}_{10}&\bold{B}_{12,11}\bold{U}_{11}&\bold{A}_{12,12}\\
\hline
\end{matrix}$}
\label{fig: example2}
\end{equation}
\hrulefill
\setcounter{equation}{4}

\end{figure*}

\begin{exam}\label{exam: exam3} Recall the DSN in \Cref{exam: exam1}. We present a coding scheme with the cooperation matrix specified in the left part of Fig.~\ref{fig: CoopMatrix}. The non-systematic part of the generator matrix is shown in (\ref{fig: example2}), which is obtained through the following process:
\begin{enumerate}
\item Partition all the non-zero-non-one elements into structured groups, each of which is marked by either a rectangle or a hexagon in $\bold{D}$, as indicated in the left part of Fig.~\ref{fig: CoopMatrix};
\item Replace the endpoints of each horizontal line segment in Step 1 with $s\in S$ ($S$ is a set of symbols), as indicated in the right part of Fig.~\ref{fig: CoopMatrix}; denote the new matrix by $\bold{X}$;
\item Assign a parameter $\gamma_s\in\Db{N}$ to each $s\in S$, and a matrix $\bold{B}'_{s}\in \textup{GF}^{k_i\times\gamma_s}$ to any $(i,j)$ such that $\bold{X}_{i,j}=s$; 
\item For each $i\in\MB{p}$, $l\in\MB{L}$, let $\eta_{i;l}=\max\nolimits_{s: k\in \Se{I}_i^l,\bold{X}_{k,i}=s} \gamma_{s}$, assign $\bold{V}_{i;l}\in\textup{GF}(q)^{\eta_{i;l}\times r_i}$ to $v_i$; let $\bold{B}_{s}=\left[\bold{B}'_s,\bold{0}_{\eta_{i;l}-\gamma_s}\right]$; assign $\bold{A}_{i,j}=\bold{B}_s\bold{V}_{j;l}$ for $s=\bold{X}_{i,j}$, $l=\bold{D}_{i,j}$.
\item Assign $\bold{A}_{i,j}$ for $\bold{X}_{i,j}=1$ according to \Cref{cons: 1}.
\end{enumerate}

Let us again focus on node $v_2$. Let $\Se{I}_2^1=\{v_1,v_3,v_5\}$, $\Se{I}_2^2=\{v_8,v_9\}$, $\Se{I}_2^3=\{v_{10},v_{11}\}$. Then, $\Se{B}_2^1=\{v_4,v_6,v_8\}$, $\Se{B}_2^2=\{v_4,v_6\}$, $\Se{B}_2^3=\varnothing$, $d_{2,0}=r_2-\delta_2-\eta_{i;2}-\eta_{i;3}$, $d_{2,1}=d_{2,0}+\delta_2+\delta_1+\delta_3+\delta_5$, $d_{2,2}=d_{2,1}+\gamma_{c}$, $d_{2,3}=d_{2,2}+\gamma_{d}$. 

We first show that knowing $\{\bold{m}_j\}_{v_j\in\Se{A}_2^1}$ is sufficient for removing $\bold{s}_2=\sum\nolimits_{j\in\Se{I}_2^1}\bold{m}_j\bold{B}_{j,2}\bold{U}_{2}+\sum\nolimits_{l=2}^{3}\sum\nolimits_{j\in\Se{I}_2^l}\bold{m}_j\bold{B}_{\bold{X}_{j,2}}\bold{V}_{2;l}$ from the parity part of $\bold{c}_2$. Note that if $\bold{A}_{i,i}$, $\bold{U}_{i}$ and $\{\bold{V}_{i;l}\}_{l\in\{2,3\}}$ are linearly independent, then for all $l$, $\sum\nolimits_{j\in\Se{I}_i^l}\bold{m}_j\bold{B}_{\bold{X}_{j,i}}$ is recoverable if $\bold{m}_i$ is recoverable. In our example, this means that $\{\bold{m}_j\bold{U}_{j,2}\}_{j=1,3,5}$, $\bold{m}_{8}\bold{B}_{m}+\bold{m}_9\bold{B}_o$, $\bold{m}_{10}\bold{B}_{o}+\bold{m}_{11}\bold{B}_x$ are known: the sum of them is exactly $\bold{s}_2$. Therefore, $\bold{s}_2$ is removed through the $1$-st level cooperation. We next show that additional parities are obtained through $l$-th level cooperations with $l=2,3$.

In the $2$-nd level cooperation, $\bold{m}_8,\bold{m}_9$ are known. Therefore, $\bold{m}_2\bold{B}_c+\bold{m}_3\bold{B}_e+\bold{m}_4\bold{B}_{\alpha}+\bold{m}_6\bold{B}_{\beta}$ is also known. We remove $\bold{m}_3\bold{B}_{e}$ that is obtained via $v_3$. In order to obtain the $\gamma_c$ parities from $\bold{m}_2\bold{B}_c$, one needs $\bold{m}_4,\bold{m}_6$ to be recoverable. Therefore, $\Se{B}_2^2=\{v_4,v_6\}$, $d_{2,2}=d_{2,1}+\gamma_{c}$, $\lambda_{2,2;\varnothing}=d_{2,1}$.
\end{exam}

As shown in the first step in \Cref{exam: exam3}, the cooperation matrix adopts a partition of non-zero-non-one elements into groups where each of them forms a cycle. Suppose there are $T$ cycles. Represent each cycle with index $t\in \MB{T}$ by a tuple $(X_t,Y_t,\{X_{t;j}\}_{j\in Y_t},\{Y_{t;i}\}_{i\in X_t},l_t)$, where $X_t$, $Y_t$ denote the indices of the rows and the columns of the cycle, respectively; $l_t$ denotes the number assigned to the vertices of the cycle; $X_{t;j}=\{i_1,i_2\}$ for $j\in Y_t$ and $(i_1,j),(i_2,j)$ are the vertices of cycle $t$; $Y_{t;i}=\{j_1,j_2\}$ for $i\in X_t$ and $(i,j_1),(i,j_2)$ are the vertices of cycle $t$. For example, let $t=1$ for the blue cycle at the bottom left part of the matrices in Fig.~\ref{fig: CoopMatrix}. Then, it is represented by $(\{10,11,12\},\{4,5,6\},\{X_{1;j}\}_{j\in\{4,5,6\}},\{Y_{1;i}\}_{i\in \{10,11,12\}},3)$, where $X_{1;4}=\{10,11\}$, $X_{1;5}=\{11,12\}$, $X_{1;6}=\{10,12\}$, $Y_{1;10}=\{4,6\}$, $Y_{1;11}=\{4,5\}$, $Y_{1;12}=\{5,6\}$.

Observe that cycle $t\in\MB{T}$ in Fig.~\ref{fig: CoopMatrix} essentially represents a pair of disconnected edges or triangles with vertices from $X_t$ and $Y_t$. We mark $X_t$, $Y_t$, draw a directed edge with label $l$ from $X_t$ to $Y_t$ for each $t\in \MB{T}$, and obtain the so-called \textbf{cooperation graph}. The cooperation graph for the coding scheme in \Cref{exam: exam3} is shown in Fig.~\ref{fig: figMHP}. 

\begin{figure}
\centering
\includegraphics[width=0.3\textwidth]{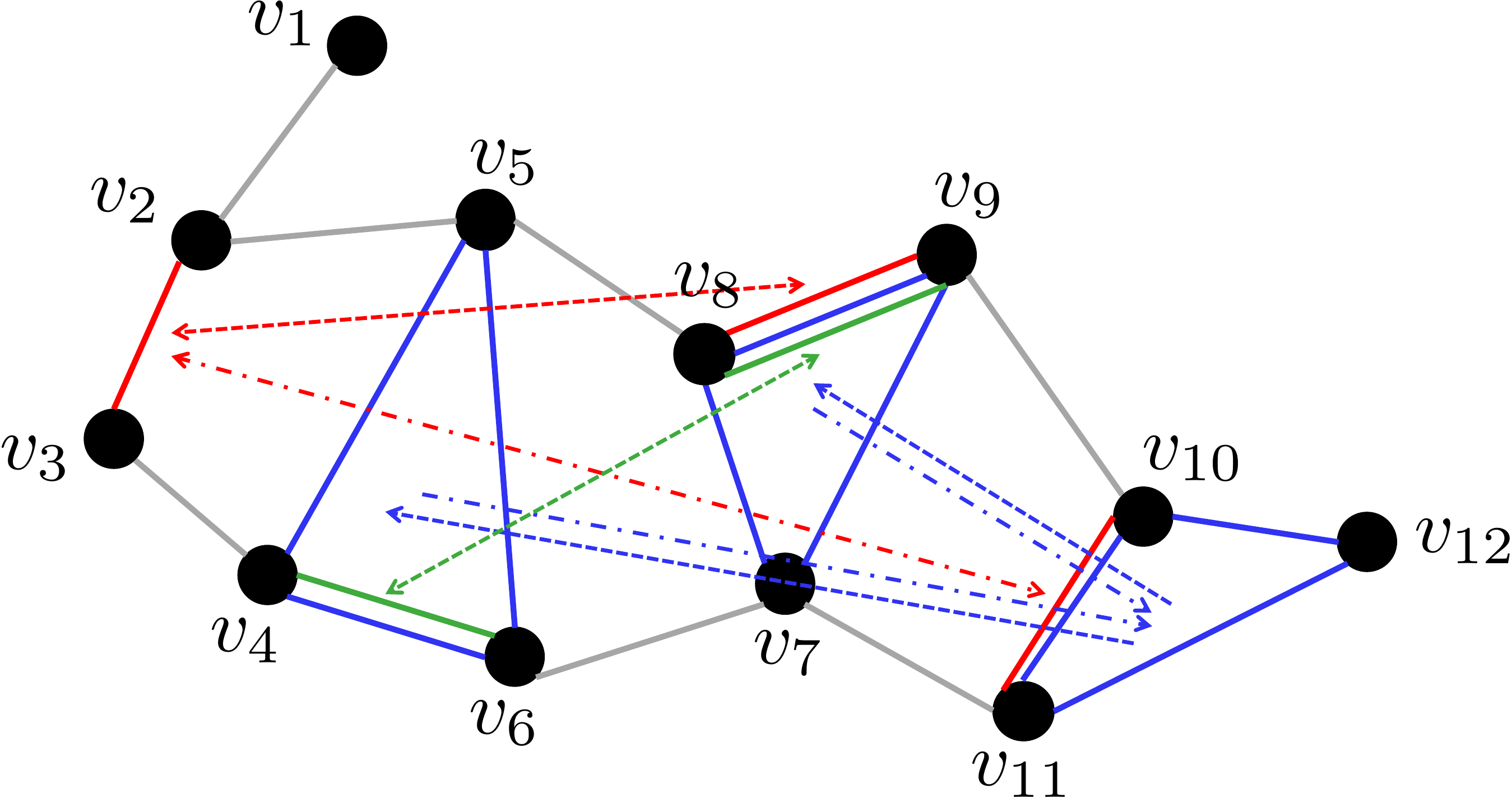}
\caption{Cooperation graph of \Cref{exam: exam3}.}
\label{fig: figMHP}
\end{figure}

\subsection{Construction Over Compatible Graphs}
\label{subsec: construction}
We have defined the notion of cooperation graphs in \Cref{subsec: cooperation graphs}. Observe that the cooperation graph shown in Fig.~\ref{fig: figMHP} satisfies a set of conditions that define the so-called \textbf{compatible graph}. We show in \Cref{theo: ECCcons2} the existence of a hierarchical coding scheme with cooperation graph $\Se{G}$ if $\Se{G}$ is a compatible graph. The coding scheme is presented in \Cref{cons: 2}.

\begin{defi}\label{defi: consistent graph} Let $\Se{G}$ be a cooperation graph on $G(V,E)$ that is represented by $\{(X_t,Y_t,\{X_{t;j}\}_{j\in Y_t},\{Y_{t;i}\}_{i\in X_t},l_t)\}_{t\in \MB{T}}$, and let the $1$-st level cooperation graph be $\Se{G}_1$. Let $\Se{M}_i$ denote the set of nodes that have an outgoing edge pointing at $v_i$ in $\Se{G}_1$. For each $v_i\in V$, $l\in \MB{L_i}$, let $R_{i;l}=\{t:i\in Y_t, t\in \MB{T},l_t=l\}$, $R_i=\bigcup\nolimits_{l\in\MB{L_i}}R_{i;l}$, $T_{i;l}=\{t:j\in R_{i;l}\cap Y_{t;i},t\in \MB{T} \}$, and $V_{i;l}=\bigcup\nolimits_{t\in R_{i;l}} Y_t$.  We call $\Se{G}$ a \textbf{compatible graph} on $G$ if the following conditions are satisfied: 
\begin{enumerate}
\item For each $v_i\in V$, sets in $\{Y_t\}_{t\in R_i}$ are disjoint;
\item For each $v_i\in V$, $l\in \MB{L_i}$, and any node $v_j$ such that $j\in V_{i;l}\setminus\{v_i\}$, $ V_{j;l}\subseteq \Se{M}_i$.
\end{enumerate}

\end{defi}

Note that in \Cref{defi: consistent graph}, $R_{i;l}$ consists of indices of all the cycles that define a level-$l$ cooperation at node $v_i$; $T_{i;l}$ consists of indices of all the cycles that provide extra parities in the $l$-th level cooperation of node $v_i$; $V_{i;l}$ consists of all the nodes that are required to be recovered in order to obtain the cross parities at $v_i$ resulting from the $l$-th level cooperation.

\begin{cons}\label{cons: 2} Let $G(V,E)$ represent a DSN with parameters $(\bold{n},\bold{k},\bold{r})$. Suppose $\Se{G}$ is a compatible graph of depth $L$ on $G$, with parameters $\{(X_t,Y_t,\{X_{t;j}\}_{j\in Y_t},\{Y_{t;i}\}_{i\in X_t},l_t)\}_{t\in \MB{T}}$, and the $1$-st level cooperation graph is denoted by $\Se{G}_1$ (other necessary parameters are as they are in \Cref{defi: consistent graph}). 

Let $\boldsymbol{\delta}$ be the $1$-st level cooperation parameter. For each $v_i\in V$, $1\leq l\leq L_i$, and any $t\in T_{i;l}$, assign cooperation parameter $\gamma_{i;t}$ to $C_t$; let $\eta_{j;l}=\max\nolimits_{t\in R_{j;l},i\in X_{t;j}} \gamma_{i;t}$.

Let $u_{i}=k_{i}+\delta_{i}+\sum\nolimits_{l=2}^{L_i}\eta_{i;l}$, $v_{i}=r_{i}+\sum\nolimits_{j\in \Se{M}_i}\delta_{j}+\sum\nolimits_{2\leq l\leq L_i,t\in T_{i;l}} \gamma_{i;t}$, for $i\in\MB{p}$. For each $i\in\MB{p}$, let $a_{i,s}, b_{i,t}$, $s\in \MB{u_{i}}$ and $t\in\MB{v_{i}}$ be distinct elements of $\textup{GF}(q)$\textcolor{edit_ah}{, where $q\geq \max\nolimits_{i\in\MB{p}}\LB{u_i+v_i}$. }

Matrix $\bold{G}$ in (\ref{eqn: GenMatDL}) is assembled as follows. Consider the Cauchy matrix $\bold{T}_{i}$ on $\textup{GF}(q)^{u_{i}\times v_{i}}$ such that $\bold{T}_{i}=\bold{Y}(a_{i,1},\dots,a_{i,u_{i}}; b_{i,1},\dots,b_{i,v_{i}})$, for $i\in\MB{p}$. Then, we obtain $\bold{A}_{i,i}$, $\bold{B}_{i,j}$, $\bold{E}_{i;l}$, $\bold{U}_{i}$, $\bold{V}_{i;l}$, for $i\in\MB{p}$, $j\in\MB{p}\setminus\LB{i}$, according to the following partition of $\bold{T}_{i}$,

\setcounter{equation}{4}
\begin{equation}\label{eqn: CRSHL}
\bold{T}_{i}=\left[
\begin{array}{c|c}
\bold{A}_{i,i} & \begin{array}{c|c|c|c}
\bold{B}_{i} & \bold{E}_{i;2} & \dots & \bold{E}_{i;L_i}
\end{array}
\\
\hline
\begin{array}{c}\bold{U}_{i}\\
\hline 
\bold{V}_{i;2}\\
\hline
\vdots\\
\hline
\bold{V}_{i;L_i}\end{array}& \bold{Z}_{i}
\end{array}\right],
\end{equation} 
\begin{equation}
\textit{where } \text{ } \bold{B}_{i}=\left[\begin{array}{c|c|c}\bold{B}_{i,j_1} & \dots & \bold{B}_{i,j_{|\Se{M}_i|}}
\end{array}
\right],
\end{equation}
\begin{equation}
\textit{and  } \text{ } \bold{E}_{i;l}=\left[\begin{array}{c|c|c}\bold{E}_{i;l;t_1} & \dots & \bold{E}_{i;l;t_{|T_{i;l}|}}
\end{array}
\right], 
\end{equation}
such that $\Se{M}_i=\{j_{1},j_{2},\dots,j_{|\Se{M}_i|\}}$, $T_{i;l}=\{t_{1},t_{2},\dots,t_{|T_{i;l}|}\}$, $\bold{A}_{i,i}\in \textup{GF}(q)^{k_{i}\times r_{i}}$, $\bold{U}_{i}\in \textup{GF}(q)^{\delta_{i}\times r_i}$, $\bold{V}_{i;l}\in \textup{GF}(q)^{\eta_{i;l}\times r_i}$, $\bold{B}_{i,j}\in \textup{GF}(q)^{k_{i}\times \delta_{j}}$ for all $v_{j}\in\Se{M}^1_i$, and $\bold{E}_{i;l;t}\in \textup{GF}(q)^{k_{i}\times \gamma_{i;t}}$. Let $\bold{B}_{i,j}=\left[ \bold{E}_{i;l;t},\bold{0}_{k_i\times(\eta_{j;l}-\gamma_{i;t})}\right]$, and $\bold{A}_{i,j}=\bold{B}_{i,j}\bold{V}_{j;l}$, for all $j\in Y_{t;i}$, $t\in T_{i;l}$; let $\bold{A}_{i,j}=\bold{B}_{i,j}\bold{U}_j$ for $v_j\in \Se{M}_i$; otherwise $\bold{A}_{i,j}=\bold{0}_{k_i\times r_i}$. Substitute components of $\bold{G}$ in (\ref{eqn: GenMatDL}) correspondingly. 
Let $\Se{C}_2$ represent the code with generator matrix $\bold{G}$.

\end{cons}

\begin{figure}
\centering
\includegraphics[width=0.22\textwidth]{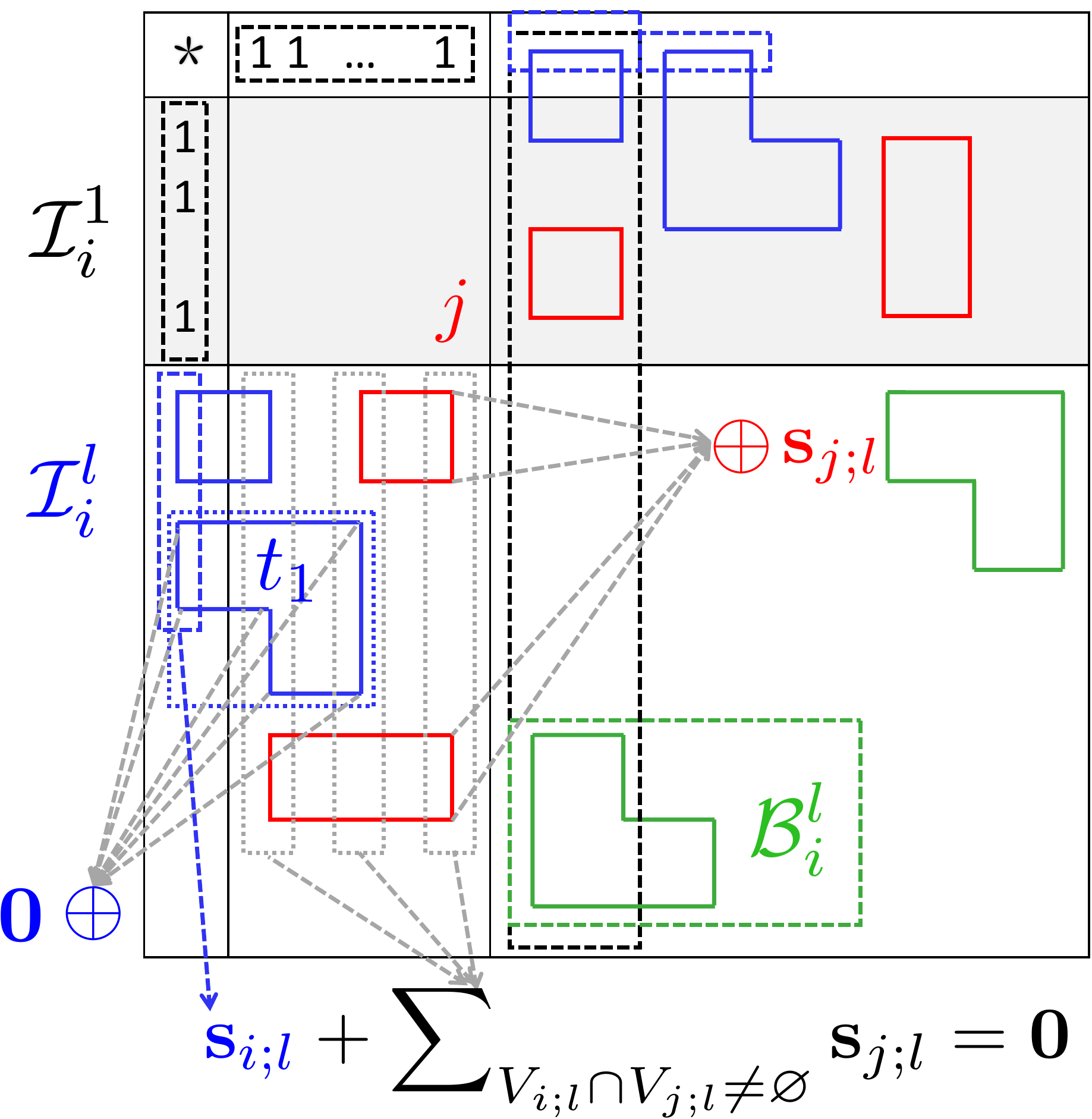}
\hspace{1pt}
\includegraphics[width=0.25\textwidth]{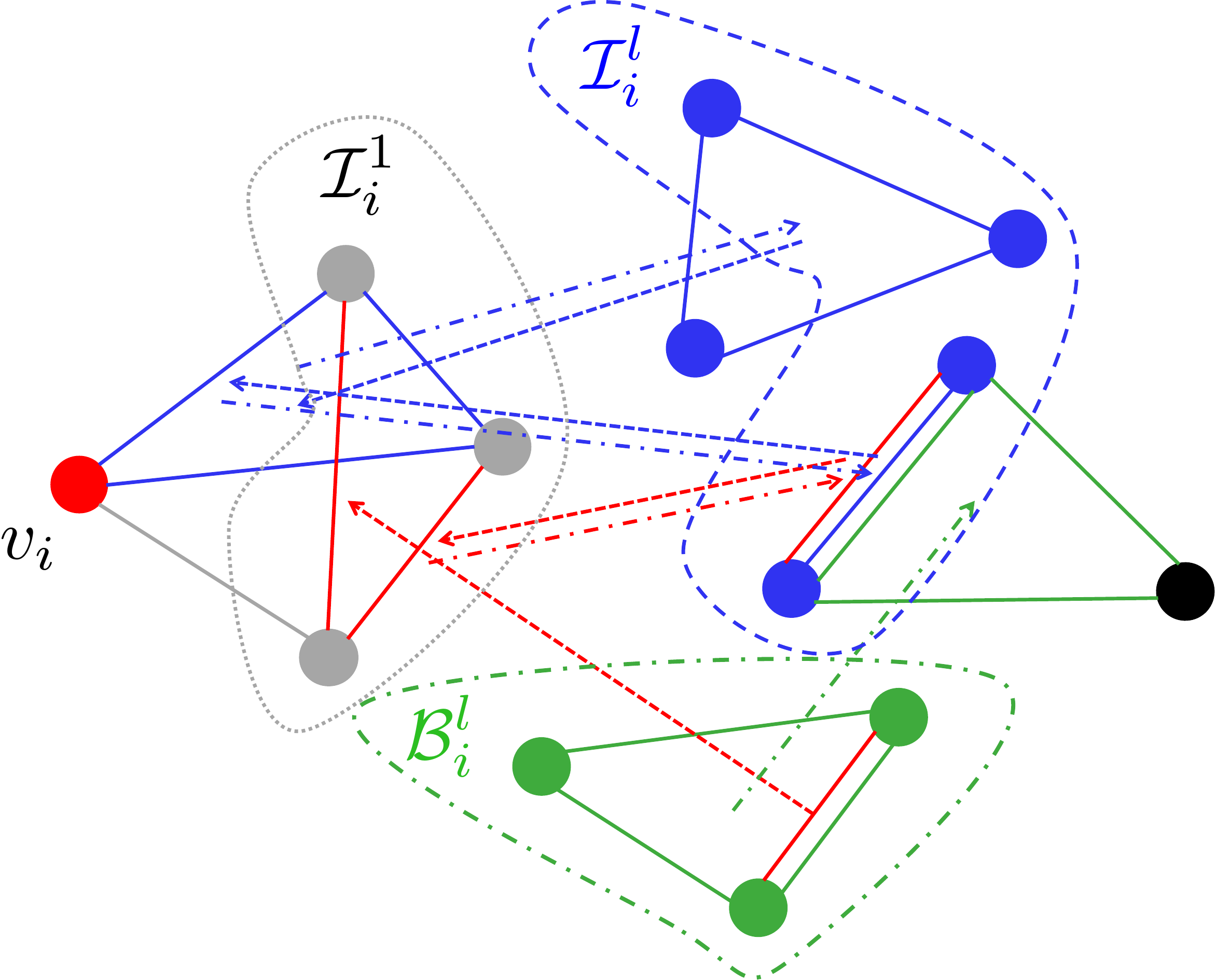}
\caption{Proof of \Cref{theo: ECCcons2}.}
\label{fig: proof}
\end{figure}

\begin{theo} \label{theo: ECCcons2} The code $\Se{C}_2$ has ECC hierarchies $\bold{d}_i=(d_{i,0},d_{i,1},\dots,d_{i,L_i})$, for all $v_i\in V$, where $d_{i,0}=r_i-\delta_i-\sum\nolimits_{l=2}^{L_i}\eta_{i;l}$, $d_{i,1}=r_i+\sum\nolimits_{j\in \Se{M}_i}\delta_j$, and $d_{i,l}=r_{i}+\sum\nolimits_{j\in \Se{M}_i}\delta_{j}+\sum\nolimits_{2\leq l'\leq l,t\in T_{i;l'}} \gamma_{i;t}$. Moreover, $\Se{I}_i^1=\Se{M}_i$, $\Se{B}_i^1=\bigcup\nolimits_{v_j\in\Se{M}_i}\textcolor{edit_ah}{\left(\Se{M}_j\setminus(\{v_i\}\cup \Se{M}_i)\right)}$; for $2\leq l\leq L_i$, $\Se{I}^l_i=\bigcup\nolimits_{t\in R_{i;l}} X_{t;i}$, $\Se{B}_i^l=\bigcup\nolimits_{v_j\in\Se{I}^l_i} \textcolor{edit_ah}{\left(\Se{I}^l_j\setminus (\{v_i\}\cup\Se{A}_i^l)\right)}$ (recall $\Se{A}_i^l=\bigcup\nolimits_{l'\leq l}\Se{I}^{l'}_i$), $\lambda_{i,l;\Se{W}}=r_i+\sum\nolimits_{v_j\in \Se{M}_i,(\Se{M}_j\setminus\{v_i\})\subseteq (\Se{M}_i\cup\Se{W})}\delta_j+\sum\nolimits_{\substack{2<l'\leq l,t\in T_{i;l}, C_t=l'', Y_{t;i}=\{j,j'\}, \\ \textcolor{edit_ah}{\Se{I}_{j}^{l''}\setminus\Se{A}_{i}^{l'}\subseteq (\{v_i\}\cup\Se{W})\allowbreak\text{ or }\Se{I}_{j'}^{l''}\setminus\Se{A}_{i}^{l'}\subseteq (\{v_i\}\cup\Se{W})}}}\gamma_{i;t}$, for $\varnothing\subseteq \Se{W}\subseteq \Se{B}_i^l$.

\end{theo}

\begin{proof} The ECC capability for the local decoding and the $1$-st level cooperation has been proved in \Cref{theo: ECCcons1}. We now focus only on the $l$-th level cooperation, where $2\leq l\leq L_i$. 

We show that for $v_i\in V$, the cross parities due to cooperation with nodes in $\Se{I}_i^l$, $2\leq l\leq L_i$, can be computed if $\{\bold{m}_j\}_{j\in \Se{I}_i^1}$ are recovered. Fig.~\ref{fig: proof} indicates a subgraph of the cooperation graph of $\Se{C}_2$ containing $v_i$, $\Se{I}_i^1$, $\Se{I}_i^l$ only (right), and its corresponding cycle representations in the cooperation matrix (left). Condition 1 in \Cref{defi: consistent graph} implies that the cycles are all disjoint; Condition 2 implies that all the cycles (marked in red) containing an edge in column $j\in V_{i;l}$ are contained within the columns representing $\Se{I}_i^1$. 

We calculate the sum of all the $l$-th level cross parities of all nodes in $\{i\}\cup \{j: V_{j;l}\cap V_{i;l}\neq \varnothing\}$ by two methods. On one hand, we first calculate the sum of all the $l$-th level cross parities at node $v_{j}$ and denote it by $\bold{s}_{j;l}$, $j\in V_{i;l}\cup\{i\}$: $\bold{s}_{j;l}$ is recoverable because $\{\bold{V}_{i;l}\}_{2\leq l\leq L}$ have linearly independent rows. Then we know that the sum is $\bold{s}_{i;l}+\sum\nolimits_{V_{i;l}\cap V_{j;l}\neq \varnothing}\bold{s}_{j;l}$. On the other hand, we first calculate the sum of the cross parities on the vertices of each cycle, which is $\bold{0}$ since the cross parities at the endpoints of each horizontal edge of any cycle are identical, and thus sum up to zero. Therefore, $\bold{s}_{i;l}+\sum\nolimits_{V_{i;l}\cap V_{j;l}\neq \varnothing}\bold{s}_{j;l}$ should also be zero, which indicates that $\bold{s}_i=\sum\nolimits_{V_{i;l}\cap V_{j;l}\neq \varnothing}\bold{s}_j$.

Moreover, for $2\leq l\leq L_i$, any cycle $C_t$ with index $t\in T_{i;l}$ has potential to provide additional $\gamma_{i;t}$ in $l$-th level cross parities to $v_i$. Therefore, $d_{i,l}=r_{i}+\sum\nolimits_{j\in \Se{M}_i}\delta_{j}+\sum\nolimits_{2\leq l'\leq l,t\in T_{i;l'}} \gamma_{i;t}$. This is done if any one of the two nodes in $Y_{t;i}=\{j,j'\}$ obtains its $\gamma_{i;t}$ cross parities at $v_i$ through its $l_t$-th level cooperation. However, the condition of $v_j$ successfully calculating these $\gamma_{i;t}$ cross parities is that all other cross parities generated from nodes lying in the $l_t$-th level cooperation are calculated and subtracted from the $l_t$-th level parities of $v_j$, i.e., $\Se{I}_j^{l_t}\setminus \Se{A}_i^{l} \subseteq \{v_i\}\cup\Se{W}$. Similarly, the condition of $v_{j'}$ successfully calculating these $\gamma_{i;t}$ cross parities at $v_i$ is described as $\Se{I}_{j'}^{l_t}\setminus \Se{A}_i^{l} \subseteq \{v_i\}\cup\Se{W}$. Therefore, the overall requirement is stated as ``$\Se{I}_j^{l_t}\setminus \Se{A}_i^{l} \subseteq \{v_i\}\cup\Se{W}$ or $\Se{I}_{j'}^{l_t}\setminus \Se{A}_i^{l} \subseteq \{v_i\}\cup\Se{W}$''. From the aforementioned discussion, we reach that $\Se{B}_i^l=\bigcup\nolimits_{t\in T_{i;l},j\in Y_{t;i}}(\Se{I}^{l_t}_j\setminus (\{v_i\}\cup\Se{A}_i^l))$ (marked in green in Fig.~\ref{fig: proof} in both left and right panels) and $\lambda_{i,l;\Se{W}}=r_i+\sum\nolimits_{v_j\in \Se{M}_i,(\Se{M}_j\setminus\{v_i\})\subseteq (\Se{M}_i\cup\Se{W})}\delta_j+\sum\nolimits_{\substack{2<l'\leq l,t\in T_{i;l'}, Y_{t;i}=\{j,j'\}, \\ \textcolor{edit_ah}{\Se{I}_{j}^{l_t}\setminus\Se{A}_{i}^{l'}\subseteq (\{v_i\}\cup\Se{W})\allowbreak\text{ or }\Se{I}_{j'}^{l_t}\setminus\Se{A}_{i}^{l'}\subseteq (\{v_i\}\cup\Se{W})}}}\gamma_{i;t}$, for $\varnothing\subseteq \Se{W}\subseteq \Se{B}_i^l$.
\end{proof}

\balance
\section{Conclusion}
\label{section: conclusion}
Hierarchical locally accessible codes in the context of centralized networks have been discussed in various prior works, whereas those of DSNs (no prespecified topology) have not been explored. In this paper, we presented a topology-aware cooperative data protection scheme for DSNs, which is an extension of, and subsumes, our previous work on hierarchical coding for centralized distributed storage. Our scheme achieves faster recovery speed compared with existing network coding methods, and corrects more erasure patterns compared with our previous work. Our work exploits the power of the blockchain technology in DSNs, and it has potential to be applied in WSNs.


\section*{Acknowledgment}
This work was supported in part by NSF under the grants CCF-BSF 1718389 and CCF 1717602, and in part by AFOSR under the grant FA 9550-17-1-0291.

\bibliography{ref}
\bibliographystyle{IEEEtran}

\end{document}